\documentclass[11pt]{article}

\usepackage{amsmath}
\usepackage{amssymb}
\usepackage{amsthm}
\usepackage{xcolor}

\usepackage[english]{babel}
\usepackage{ifthen}
\usepackage{booktabs}
\usepackage{multirow}
\usepackage{xspace}
\usepackage{url}
\usepackage[final]{microtype}
\usepackage{enumitem}
\usepackage{wrapfig}

\usepackage{setspace}

\usepackage{tikz}
\usetikzlibrary{shapes}
\usetikzlibrary{calc}
\usetikzlibrary{decorations.pathmorphing}

\newenvironment{mk}{\noindent\color{blue} MK:} {}

\newcommand{\reflem}[1]{Lemma~\ref{#1}}
\newcommand{\refprp}[1]{Proposition~\ref{#1}}

\newcommand{\reffig}[1]{Figure~\ref{#1}}

\newcommand{\set}[2]{\left\{#1\mathrel{\left|\vphantom{#1}\vphantom{#2}\right.}#2\right\}}
\newcommand{\oneset}[1]{\left\{\mathinner{#1}\right\}}

\newcommand{\abs}[1]{\left|\mathinner{#1}\right|}

\newcommand{\N}{\mathbb{N}}

\newcommand{\ltrue}      {\ensuremath{\mathord{\top}}\xspace}
\newcommand{\lfalse}     {\ensuremath{\mathord{\bot}}\xspace}

\font\petite=cmmi10 at 8pt
\def\malcev{\mathbin{\hbox{$\bigcirc$\rlap{\kern-9pt\raise0,75pt\hbox{\petite m}}}}}

\newcommand{\Alpha}{\mathsf{alph}}

\newcommand{\varietyFont}[1]{\mathrm{\mathbf{#1}}}
\newcommand{\DA}{\varietyFont{D\hspace{-1pt}A}}
\newcommand{\LI}{\varietyFont{LI}}

\newcommand{\Jone}{\varietyFont{J}_1}

\newcommand{\Ap}{\varietyFont{A}}

\newcommand{\Jtrivial}{\varietyFont{J}}

\let\J\Jtrivial
\let\V\Vara

\newcommand{\VarFO}{\varietyFont{FO}}

\newcommand{\K}{\varietyFont{K}}
\newcommand{\D}{\varietyFont{D}}
\newcommand{\R}[1]{\varietyFont{R}_{#1}}
\let\L=\undefined
\newcommand{\L}[1]{\varietyFont{L}_{#1}}

\newcommand{\join}{\mathrel{\vee}}
\newcommand{\ord}{\mathrm{ord}}

\newcommand{\greenfont}[1] {\ensuremath{\mathcal{#1}}}

\newcommand{\greenR} {\greenfont{R}\xspace}
\newcommand{\Req}    {\mathrel{\greenR}}
\newcommand{\Rleq}   {\mathrel{\leq_\greenR}}

\newcommand{\Rg}     {\mathrel{>_\greenR}}

\newcommand{\greenL} {\greenfont{L}\xspace}
\newcommand{\Leq}    {\mathrel{\greenL}}
\newcommand{\Lleq}   {\mathrel{\leq_\greenL}}
\newcommand{\Ll}     {\mathrel{<_\greenL}}

\newcommand{\greenJ} {\greenfont{J}\xspace}
\newcommand{\Jeq}    {\mathrel{\greenJ}}
\newcommand{\Jleq}   {\mathrel{\leq_\greenJ}}
\newcommand{\Jl}     {\mathrel{<_\greenJ}}

\newcommand{\Right}{\mathrel{\triangleright}}
\newcommand{\Left}{\mathrel{\triangleleft}}

\newcommand{\complexityfont}[1]{\ensuremath{\mathrm{#1}}\xspace}

\newcommand{\NP}    {\complexityfont{NP}}

\newcommand{\logicfont}[1]{\mathrm{#1}}

\newcommand{\FO}{\logicfont{FO}}

\newcommand{\suc}{\ensuremath{{+}1}\xspace}

\newcommand{\X}{\mathop{\mathsf{X}\vphantom{b}}\nolimits}
\newcommand{\Y}{\mathop{\mathsf{Y}\vphantom{b}}\nolimits}
\newcommand{\Z}{\mathop{\mathsf{Z}\vphantom{b}}\nolimits}
\newtheorem{theorem}{Theorem}%[section]

\newtheorem{proposition}[theorem]{Proposition}
\newtheorem{lemma}[theorem]{Lemma}
\newtheorem{corollary}[theorem]{Corollary}
\newtheorem{conjecture}[theorem]{Conjecture}
\newtheorem{pro-remark}[theorem]{Remark}
\newenvironment{remark}{\begin{pro-remark}\rm}{\end{pro-remark}}

\renewenvironment{proof}[1][]{\pagebreak[3]\noindent\textit{Proof\ifthenelse{\equal{#1}{}}{}{ (#1)}. }}{\pagebreak[3]\medskip}
\newtheorem{expl}[theorem]{Example}

\setlist{itemsep=1pt,parsep=0pt,topsep=2pt}

\let\phi\varphi
\let\calV\LangV
\def\dast{\mathbin{\ast\ast}}
\def\calP{\mathcal{P}}

\begin{document}

\title{The $\FO^2$ alternation hierarchy is decidable}

\author{Manfred Kuf\-leitner\thanks{The
    first author was supported by the German Research Foundation
    (DFG) under grant \mbox{DI 435/5-1}.} \\
  {\small University of Stuttgart, Germany} \\
%  Universit{\"a}tsstr.\ 38, 70569 Stuttgart, Germany \\
  {\small \texttt{kufleitner{@}fmi.uni-stuttgart.de}}
  \and Pascal Weil\thanks{The second author was supported by
    the grant ANR 2010 BLAN 0202 01 FREC.}\\
  {\small CNRS, LaBRI, UMR5800, F-33400 Talence, France} \\
  {\small Univ. Bordeaux, LaBRI, UMR5800, F-33400 Talence, France} \\%[-1mm]
  {\small \texttt{pascal.weil{@}labri.fr}}
} 

\date{}

\maketitle

\begin{abstract}
  \noindent
  We consider the two-variable fragment $\FO^2[{<}]$ of first-order
  logic over finite words. Numerous characterizations of this class
  are known. Th{\'e}\-rien and Wilke have shown that it is decidable
  whether a given regular language is definable in $\FO^2[{<}]$.  From
  a practical point of view, as shown by Weis, $\FO^2[{<}]$ is
  interesting since its satisfiability problem is in
  $\NP$. Restricting the number of quantifier alternations yields an
  infinite hierarchy inside the class of $\FO^2[{<}]$-definable
  languages. We show that each level of this hierarchy is decidable.
  For this purpose, we relate each level of the hierarchy with a
  decidable variety of finite monoids.

  Our result implies that there are many different ways of climbing up
  the $\FO^2[{<}]$-quantifier alternation hierarchy: deterministic and
  co-deterministic products, Mal'cev products with definite and
  reverse definite semigroups, iterated block products with $\greenJ$-trivial monoids, and some inductively defined omega-term identities. A
  combinatorial tool in the process of ascension is that of
  condensed rankers, a refinement of the rankers of Weis and Immerman
  and the turtle programs of Schwentick, Th{\'e}rien, and Vollmer.
\end{abstract}

\pagebreak

%%%%%%%%%%%%
\section{Introduction}

The investigation of logical fragments has a long history.  McNaughton
and Papert~\cite{mp71:short} showed that a language over finite words
is definable in first-order logic $\FO[{<}]$ if and only if it is
star-free. Combined with Sch{\"u}tzenberger's characterization of
star-free languages in terms of finite aperiodic
monoids~\cite{sch65sf:short}, this leads to an algorithm to decide
whether a given regular language is first-order definable. Many other
characterizations of this class have been given over the past 50
years, see~\cite{dg08SIWT:short} for an overview.  Moreover, mainly
due to its relation to linear temporal logic~\cite{kam68:short}, it
became relevant to a large number of application fields, such as
verification.

Very often one is interested in fragments of first-order logic. From a
practical point of view, the reason is that smaller fragments often
yield more efficient algorithms for computational problems such as
satisfiability. For example, satisfiability for $\FO[{<}]$ is
non-elementary~\cite{Sto74:short}, whereas the satisfiability problem
for first-order logic with only two variables is in $\NP$,
cf.~\cite{wei11phd}.  And on the theoretical side, fragments form the
basis of a descriptive complexity theory inside the regular
languages: the simpler a logical formula defining a language, the
easier the language.  Moreover, in contrast to classical
complexity theory, in some cases one can actually decide whether a
given language has a particular property.  From both the practical and the
theoretical point of view, several natural hierarchies have been
considered in the literature: the quantifier alternation hierarchy
inside $\FO[{<}]$ which coincides with the Straubing-Th{\'e}rien
hierarchy~\cite{str81tcs:short,the81tcs:short}, the quantifier
alternation hierarchy inside $\FO[{<},\suc]$ with a successor
predicate $\suc$ which coincides with the dot-depth
hierarchy~\cite{cb71jcss:short,tho82:short}, the until hierarchy of
temporal logic~\cite{tw01}, and the until-since hierarchy~\cite{tw04}.
Decidability is known for the levels of the until and the
the until-since hierarchies, and only for the very first
levels of the alternation hierarchies, see
e.g.~\cite{dgk08ijfcs:short,pin97handbook}.

%\enlargethispage{\baselineskip}

Fragments are usually defined by restricting resources in a
formula. Such resources can be the predicates which are allowed, the
quantifier depth, the number of quantifier alternations, or the number
of variables. When the quantifier depth is restricted, only finitely
many languages are definable over a fixed alphabet: decidability of
the membership problem is not an issue in this case. When restricting
the number of variables which can be used (and reused), then
first-order logic $\FO^3[{<}]$ with three variables already has the
full expressive power of $\FO[{<}]$,
see~\cite{ik89iandc,kam68:short}. On the other hand, first-order logic
$\FO^2[{<}]$ with only two variables defines a proper subclass. The
languages definable in $\FO^2[{<}]$ have a huge number of different
characterizations, see
e.g.~\cite{dgk08ijfcs:short,tt02:short,tt07lmcs:short}. For example,
$\FO^2[{<}]$ has the same expressive power as $\Delta_2[{<}]$; the
latter is a fragment of $\FO[{<}]$ with two blocks of
quantifiers~\cite{tw98stoc:short}.

\emph{Turtle programs} are one of these numerous descriptions of
$\FO^2[{<}]$-defin\-able languages~\cite{stv01dlt:short}. They are
sequences of instructions of the form ``go to the next $a$-position''
and ``go to the previous $a$-position''. Using the term \emph{ranker}
for this concept and having a stronger focus on the order of positions
defined by such sequences, Weis and Immerman~\cite{wi09lmcs}
were able to give a combinatorial characterization of the alternation
hierarchy $\FO^2_m[{<}]$ inside $\FO^2[{<}]$. Straubing~\cite{str11csl}
gave an algebraic characterization of $\FO^2_m[{<}]$. But neither result yields the decidability of $\FO^2_m[{<}]$-definability for $m > 2$. In some sense, this is
the opposite of a previous result of the
authors~\cite[Thm. 6.1]{kw11LMCS}, who give necessary and
sufficient conditions which helped to decide the
$\FO^2_m[{<}]$-hierarchy with an error of at most one.  In this paper
we give a new algebraic characterization of $\FO^2_m[{<}]$, and
this characterization immediately yields decidability.

The algebraic approach to the membership problem of logical fragments
has several advantages. In favorable cases, it opens the road to decidability procedures.  Moreover, it allows a more \emph{semantic} comparison of
fragments; for example, the equality $\FO^2[{<}] = \Delta_2[{<}]$ was
obtained by showing that both $\FO^2[{<}]$ and $\Delta_2[{<}]$
correspond to the same variety of finite
monoids, namely $\DA$~\cite{pw97:short,tw98stoc:short}.

Building on previous detailed knowledge of the lattice of \emph{band}
varieties (varieties of idempotent monoids), Trotter and Weil defined
a sub-lattice of the lattice of subvarieties of $\DA$~\cite{tw97au},
which we call the {$\R{m}$-$\L{m}$-hierarchy}. These varieties have
many interesting properties and in particular, each $\R{m}$
(resp.~$\L{m}$) is efficiently decidable (by a combination of results
of Trotter and Weil \cite{tw97au}, Kuf\-leitner and Weil \cite{kw10sf},
and Straubing and Weil \cite{2012:StraubingWeil}, see
Section~\ref{sec: main section} for more details). Moreover, one can
climb up the $\R{m}$-$\L{m}$-hierarchy algebraically, using Mal'cev
products, see~\cite{kw10sf} and Section~\ref{sec: preliminaries}
below; language-theoretically, in terms of alternated closures under
deterministic and co-deterministic products~\cite{1980:Pin,kw11LMCS}; and
combinatorially using \emph{condensed} rankers,
see~\cite{kw09mfcs:short,lps08ifiptcs} and Section~\ref{sec:
  preliminaries}.

We relate the $\FO^2[{<}]$ quantifier alternation hierarchy with the
$\R{m}$-$\L{m}$-hierarchy. More precisely, the main result of this
paper is that a language is definable in $\FO^2_m[{<}]$ if and only if
it is recognized by a monoid in $\R{m+1} \cap \L{m+1}$, thus
establishing the decidability of each $\FO^2_m[{<}]$. This result was
first conjectured in~\cite{kw09mfcs:short}, where one inclusion was
established.  Our proof combines a technique introduced by
Kl{\'i}ma~\cite{kli11dm:short} and a substitution
idea~\cite{kl11lics:short} with algebraic and combinatorial tools
inspired by~\cite{kw11LMCS}. The proof is by induction and the base
case is Simon's Theorem on piecewise testable languages~\cite{sim75}.

%%%%%%%%%%%%
\section{Preliminaries}\label{sec: preliminaries}

Let $A$ be a finite alphabet and let $A^*$ be the set of all finite
words over~$A$.  The \emph{length} $\abs{u}$ of a word $u = a_1 \cdots
a_n$, $a_i \in A$, is $n$ and its \emph{alphabet} is $\Alpha(u) =
\oneset{a_1, \ldots, a_n} \subseteq A$. A position $i$ of $u = a_1
\cdots a_n$ is an \emph{$a$-position} if $a_i = a$.  A factorization
$u = u_{-} a u_{+}$ is the \emph{$a$-left factorization} of $u$ if $a
\not\in \Alpha(u_{-})$, and it is the \emph{$a$-right factorization}
if $a \not\in \Alpha(u_{+})$, i.e., we factor at the first or at the
last $a$-position.

%%%%%%%%%%%%
\subsection{Rankers}

A \emph{ranker} is a nonempty word over the alphabet
$\set{\X_a,\Y_a}{a \in A}$. It is interpreted as a sequence of
instructions of the form ``go to the next $a$-position'' and ``go to
the previous $a$-position''. More formally, for $u = a_1 \cdots a_n
\in A^*$ and $x \in \oneset{0,\ldots,n+1}$ we let
\begin{align*}
  \X_a(u, x) &= \min\set{y}{y > x \text{ and } a_y = a}, %
  \quad %
  &\X_a(u) &= \X_a(u, 0),%
  \\
  \Y_a(u, x) &= \max\set{y}{y < x \text{ and } a_y = a}, %
  \quad %
  &\Y_a(u) &= \Y_a(u, n+1).
\end{align*}
Here, both the minimum and the maximum of the empty set are undefined.
The modality $\X_a$ is for ``ne$\X$t-$a$'' and $\Y_a$ is for
``$\Y$esterday-$a$''.  For $r = \Z s$, $\Z \in \set{\X_a,\Y_a}{a\in
  A}$, we set
\begin{align*}
  r(u,x) &= s(u,\Z(u,x)), & r(u) &= s(u,\Z(u)).
\end{align*}
In particular, rankers are executed (as a set of instructions) from
left to right.
Every ranker $r$ either defines a unique position in
a word $u$, or it is undefined on $u$. For example, $\X_a \Y_b
\X_c(bca) = 2$ and $\X_a \Y_b \X_c(bac) = 3$ whereas $\X_a \Y_b
\X_c(cabc)$ and $\X_a \Y_b \X_c(bcba)$ are undefined.  A ranker $r$ is
\emph{condensed} on $u$ if it is defined and, during the execution of
$r$, no previously visited position is
overrun~\cite{kw11LMCS}. One can think of condensed rankers
as \emph{zooming in} on the position they define, see Figure~\ref{fig: condensed}.
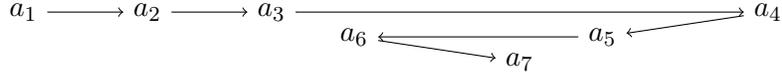
\begin{figure}\centering
    \begin{tikzpicture}[scale=0.11]
      \draw (0.0,-3.0) node (n1) {$a_1$};
      \draw (15.0,-3.0) node (n2) {$a_2$};
      \draw (30.0,-3.0) node (n3) {$a_3$};
      \draw (90.0,-3.0) node (n4) {$a_4$};
      \draw (70.0,-6.0) node (n5) {$a_5$};
      \draw (40.0,-6.0) node (n6) {$a_6$};
      \draw (60.0,-9.0) node (n7) {$a_7$};
	
      \draw[->] (n1) -- (n2);
      \draw[->] (n2) -- (n3);
      \draw[->] (n3) -- (n4);
      \draw[->] (n4) -- (n5);
      \draw[->] (n5) -- (n6);
      \draw[->] (n6) -- (n7);
    \end{tikzpicture}
    \caption{The positions defined by $r$ in $u$, when $r = \mathsf{X}_{a_1} \mathsf{X}_{a_2} \mathsf{X}_{a_3} \mathsf{X}_{a_4} \mathsf{Y}_{a_5} \mathsf{Y}_{a_6} \mathsf{X}_{a_7}$ is condensed on $u$}\label{fig: condensed}
\end{figure}
More formally $r = \Z_1 \cdots \Z_k$, $\Z_i \in
\set{\X_a,\Y_a}{a\in A}$, is \emph{condensed} on $u$ if there exists a
chain of open intervals
\begin{equation*}
  (0;\abs{u}+1) = (x_0; y_0) \supset (x_1;y_1) \supset \cdots \supset
  (x_{n-1};y_{n-1}) \ni r(u)
\end{equation*}
such that for all $1 \leq \ell \leq n-1$ the following properties are
satisfied:
\begin{itemize}
\item If $\Z_{\ell} \Z_{\ell + 1} = \X_a \X_b$, then
  $(x_\ell;y_\ell) = (\X_a(u,x_{\ell-1}); y_{\ell-1})$.
\item If $\Z_{\ell} \Z_{\ell + 1} = \Y_a \Y_b$, then
  $(x_\ell;y_\ell) = (x_{\ell-1}; \Y_a(u,y_{\ell-1})$.
\item If $\Z_{\ell} \Z_{\ell+1} = \X_a \Y_b$, then
  $(x_\ell;y_\ell) = (x_{\ell-1}; \X_a(u,x_{\ell-1}))$.
\item If $\Z_{\ell} \Z_{\ell+1} = \Y_a \X_b$, then
  $(x_\ell;y_\ell) = (\Y_a(u,y_{\ell-1}); y_{\ell-1})$.
\end{itemize}
\pagebreak[3]
For example, $\X_a \Y_b \X_c$ is condensed on $bca$ but not on $bac$.

The \emph{depth} of a ranker is its length as a word.  A \emph{block}
of a ranker is a maximal factor of the form $\X_{a_1} \cdots \X_{a_k}$
or of the form $\Y_{b_1} \cdots \Y_{b_\ell}$. A ranker with $m$
blocks changes direction $m-1$ times. By $R_{m,n}$ we denote the
class of all rankers with depth at most $n$ and with up to $m$ blocks.
We write $R^{\X}_{m,n}$ for the set of all rankers in $R_{m,n}$ which
start with an $\X_a$-modality and we write $R^{\Y}_{m,n}$ for all
rankers in $R_{m,n}$ which start with a $\Y_a$-modality.

We define $u \Right_{m,n} v$ if the same
rankers in $R^{\X}_{m,n} \cup R^{\Y}_{m-1,n-1}$ are condensed on $u$ and $v$. Similarly, $u
\Left_{m,n} v$ if the same rankers in
$R^{\Y}_{m,n} \cup R^{\X}_{m-1,n-1}$ are condensed on $u$ and $v$. The relations $\Right_{m,n}$ and
$\Left_{m,n}$ are finite index congruences \cite[Lem.~3.13]{kw11LMCS}.

The \emph{order type} $\ord(i,j)$ is one of $\oneset{{<},{=},{>}}$,
depending on whether $i<j$, $i=j$, or $i>j$, respectively. We define
$u \equiv_{m,n} v$ if
\begin{itemize}
\item the same rankers in $R_{m,n}$ are defined on $u$ and $v$,
\item for all $r \in R^{\X}_{m,n}$ and $s \in R^{\Y}_{m,n-1}$:\;
  $\ord(r(u),s(u)) = \ord(r(v),s(v))$,
\item for all $r \in R^{\Y}_{m,n}$ and $s \in R^{\X}_{m,n-1}$:\;
  $\ord(r(u),s(u)) = \ord(r(v),s(v))$,
\item for all $r \in R^{\X}_{m,n}$ and $s \in R^{\X}_{m-1,n-1}$:\;
  $\ord(r(u),s(u)) = \ord(r(v),s(v))$,
\item for all $r \in R^{\Y}_{m,n}$ and $s \in R^{\Y}_{m-1,n-1}$:\;
  $\ord(r(u),s(u)) = \ord(r(v),s(v))$.
\end{itemize}

\begin{remark}\label{remark J}
For $m = 1$, each of the families
$(\equiv_{1,n})_n$, $(\Right_{1,n})_n$, and $(\Left_{1,n})_n$ defines the class of
piecewise testable languages, see
e.g.~\cite{kli11dm:short,sim75}. Recall that a language $L \subseteq A^*$ is \emph{piecewise testable} if it is a Boolean combination of languages of the form $A^*a_1A^* \cdots a_kA^*$ ($k\ge 0$, $a_1,\ldots,a_k \in A$).
\end{remark}

%%%%%%%%%%%%
\subsection{First-order Logic}

We denote by $\FO[<]$ the first-order logic over words interpreted as
labeled linear orders. The atomic formulas are $\ltrue$ (for
\emph{true}), $\lfalse$ (for \emph{false}), the unary predicates
$\mathbf{a}(x)$ (one for each $a\in A$), and the binary predicate $x <
y$ for variables $x$ and $y$. Variables range over the linearly
ordered positions of a word and $\mathbf{a}(x)$ means that $x$ is
an $a$-position.
Apart from the Boolean connectives, we allow composition of formulas
using existential quantification $\exists x\colon \varphi$ and
universal quantification $\forall x\colon \varphi$ for $\varphi \in
\FO[{<}]$. The semantics is as usual. A \emph{sentence} in $\FO[{<}]$ is a
formula without free variables. For a sentence $\varphi$ the
\emph{language defined by $\varphi$}, denoted by $L(\varphi)$, is the
set of all words $u \in A^*$ which model $\varphi$.

The fragment \emph{$\FO^2[{<}]$} of first-order logic consists of all
formulas which use at most two different names for the variables.
This is a natural restriction, since $\FO$ with three variables
already has the full expressive power of $\FO$.  A formula $\varphi
\in \FO^2[{<}]$ is in $\FO^2_{m}[{<}]$ if, on every path of its parse
tree, $\varphi$ has at most $m$ blocks of alternating quantifiers.

Note that $\FO^2_1[{<}]$-definable languages are exactly the piecewise testable languages, cf.~\cite{str11csl}. 
For $m\ge 2$, we rely on the following important result, due to Weis and Immerman~\cite[Thm.~4.5]{wi09lmcs}.

\begin{theorem}\label{WI 2009}
A language $L$ is definable in $\FO^2_m[{<}]$ if and only if there exists $n \in \N$ such that $L$ is a union of $\equiv_{m,n}$-classes.
\end{theorem}

\begin{remark}
The definition of $\equiv_{m,n}$ above is formally different from the conditions in Weis and Immerman's \cite[Thm.~4.5]{wi09lmcs}. A careful but elementary examination reveals that they are actually equivalent.
\end{remark}

%%%%%%%%%%%%
\subsection{Algebra}

A monoid $M$ \emph{recognizes} a language $L \subseteq A^*$ if there
exists a morphism $\varphi : A^* \to M$ such that $L =
\varphi^{-1}\varphi(L)$. If $\varphi : A^* \to M$ is a morphism,
then we set $u \equiv_\varphi v$ if $\varphi(u) = \varphi(v)$. The
join ${\equiv_1} \join {\equiv_2}$ of two congruences $\equiv_1$ and
$\equiv_2$ is the least congruence containing $\equiv_1$
and~$\equiv_2$. An element $u$ is \emph{idempotent} if $u^2 = u$. The
set of all idempotents of a monoid $M$ is denoted by $E(M)$.  For
every finite monoid~$M$ there exists $\omega \in \N$ such that
$u^\omega$ is idempotent for all $u \in M$. \emph{Green's relations}
$\greenJ$, $\greenR$, and $\greenL$ are an important concept to
describe the structural properties of a monoid $M$: we set $u \Jleq v$
(resp.\ $u \Rleq v$, $u \Lleq v$) if $u = pvq$ (resp.\ $u = vq$, $u =
pv$) for some $p,q \in M$. We also define $u \Jeq v$ (resp.\ $u \Req
v$, $u \Leq v$) if $u \Jleq v$ and $v \Jleq u$ (resp.\ $u \Rleq v$ and
$v \Rleq u$, $u \Lleq v$ and $v \Lleq u$).  A monoid $M$ is
\emph{$\greenJ$-trivial} (resp.\ \emph{$\greenR$-trivial}, \emph{$\greenL$-trivial})
if $\greenJ$ (resp.\ $\greenR$, $\greenL$) is the identity relation on
$M$. We define the relations $\sim_\K$, $\sim_\D$, and $\sim_\LI$ on $M$ as follows:
\begin{itemize}
\item $u \sim_\K v$ if and only if, for all $e\in E(M)$, we have either $eu,ev \Jl e$,
  or $eu = ev$.
\item $u \sim_\D v$ if and only if, for all $f \in E(M)$, we have either $uf,vf \Jl f$,
  or $uf = vf$.
\item $u \sim_\LI v$ if and only if, for all $e,f \in E(M)$ such that $e\Jeq f$, we have either $euf,evf \Jl e$,
  or $euf = evf$.
\end{itemize}
The relations $\sim_\K$, $\sim_\D$ and $\sim_\LI$ are
congruences~\cite{krt68arbib8}. If $\V$ is a class of finite monoids, we say that a monoid $M$ is in $\K\malcev \V$ (resp.\ $\D\malcev\V$, $\LI\malcev\V$) if $M/{\sim_\K} \in \V$ (resp.\ $M/{\sim_\D} \in \V$, $M/{\sim_\LI} \in \V$). The classes $\K\malcev \V$, $\D\malcev \V$ and  $\LI\malcev \V$ are called \emph{Mal'cev products} and they are usually defined in terms of relational morphisms. In the present context however, the definition above will be sufficient~\cite{krt68arbib8}, see \cite{hw99sf}.
We will need the following classes of finite monoids:
\begin{itemize}
\item $\Jone$ consists of all finite commutative monoids satisfying
  $x^2 = x$.
\item $\Jtrivial$ (resp.\ $\R{}$, $\L{}$) consists of all finite
  $\greenJ$-trivial (resp.\ $\greenR$-trivial, $\greenL$-trivial)
  monoids.
\item $\Ap$ consists of all finite monoids satisfying $x^{\omega+1} =
  x^\omega$. Monoids in $\Ap$ are called \emph{aperiodic}.
\item $\DA$ consists of all finite monoids satisfying $(xy)^\omega x
  (xy)^\omega = (xy)^\omega$.
\item $\R{1} = \L{1} = \Jtrivial$, \ $\R{m+1} = \K \malcev \L{m}$, \
  $\L{m+1} = \D \malcev \R{m}$.
\end{itemize}
It is well known that
\begin{gather*}
  \DA = \LI \malcev \Jone, \ 
  \R{2} = \R{}, \
  \L{2} = \L{}, \ 
  \R{} \cap \L{} = \Jtrivial, \text{and} \\
  \R{m} \cup \L{m} \;\subseteq\; \R{m+1} \cap \L{m+1} 
  \;\subset\; \DA \;\subset\; \Ap
\end{gather*}
see e.g.~\cite{pin86}. The $\R{m}$-$\L{m}$-hierarchy is depicted in
\reffig{fig:RmLm}.

\begin{figure}%[thb]%{r}{5.4cm}
  \centering
  \begin{tikzpicture}
      % Node positions.
      \draw 
           (0, 0)    node[inner sep=0pt] (Jone) {$\bullet$}
      ++(   0, 0.75)  node[inner sep=0pt] (J) {$\bullet$}
      ++(-1.0, 0.5)  node[inner sep=0pt] (Rtriv) {$\bullet$}
      ++( 2.0, 0.0)  node[inner sep=0pt] (Ltriv) {$\bullet$}
      ++(-1.0, 0.5)  node[inner sep=0pt] (RjoinL) {$\bullet$}
      ++( 0.0, 0.75) node[inner sep=0pt] (RtcapLt) {$\bullet$}
      ++(-1.0, 0.5)  node[inner sep=0pt] (Rthree) {$\bullet$}
      ++( 2.0, 0.0)  node[inner sep=0pt] (Lthree) {$\bullet$}

      ++(-1.0, 0.5)  node[inner sep=0pt] (RtjoinLt) {$\bullet$}
      ++( 0.0, 0.75) node[inner sep=0pt] (RfcapLf) {$\bullet$}
      ++(-1.0, 0.5)  node[inner sep=0pt] (Rfour) {$\bullet$}
      ++( 2.0, 0.0)  node[inner sep=0pt] (Lfour) {$\bullet$}
      ++(-1.0, 0.5)  node[inner sep=10pt] (RvjoinLv) {}

      ++ (0.0, -0.3)  node[anchor=south,inner sep=10pt] (dots) {$\vdots$}
      ++ (0.0, 0.9)  node[anchor=south,inner sep=0pt] (DA) {$\DA$};
      
      % Label the nodes.
      \draw (Jone) node[below,outer sep=1mm] {$\R{1} = \L{1} = \Jtrivial$}
            (J)  node[left,inner sep=0mm,outer sep=0mm,anchor=north east] {$\J = \R{2} \cap \L{2}${}$\,$}
            (Rtriv)    node[left,outer sep=1mm,anchor=east] {$\R{}=\R{2}$}
            (Ltriv)    node[right,outer sep=1mm,anchor=west] {$\L{2}=\L{}$}
            (RjoinL)   node[right,inner sep=0mm,outer sep=0mm,anchor=south west] {\ $\R{2}\vee\L{2}$}
            (RtcapLt)  node[left,inner sep=0mm,outer sep=0mm,anchor=north east] {$\R{3}\cap\L{3}${}$\,$}
            (Rthree)   node[left,outer sep=1mm,anchor=east] {$\R{3}$}
            (Lthree)   node[right,outer sep=1mm,anchor=west] {$\L{3}$}
            (Rfour)    node[left,outer sep=1mm,anchor=east] {$\R{4}$}
            (Lfour)    node[right,outer sep=1mm,anchor=west] {$\L{4}$}
            (RtjoinLt) node[right,inner sep=0mm,outer sep=0mm,anchor=south west] {\ $\R{3}\vee\L{3}$}
            (RfcapLf)  node[left,inner sep=0mm,outer sep=0mm,anchor=north east] {$\R{4}\cap\L{4}${}$\,$};
      
      % Connections between nodes.
      \draw (Jone.center) -- (J.center)
            (J.center) -- (Rtriv.center) -- (RjoinL.center)
            (J.center) -- (Ltriv.center) -- (RjoinL.center)
            (RjoinL.center) -- (RtcapLt.center)
            (RtcapLt.center) -- (Rthree.center) -- (RtjoinLt.center)
            (RtcapLt.center) -- (Lthree.center) -- (RtjoinLt.center)
            (RtjoinLt.center) -- (RfcapLf.center)
            (RfcapLf.center) -- (Rfour.center) -- (RvjoinLv)
            (RfcapLf.center) -- (Lfour.center) -- (RvjoinLv);
  \end{tikzpicture}
  \caption{The $\R{m}$-$\L{m}$-hierarchy}\label{fig:RmLm}
\end{figure}

%%%%%%%%%%%%
\subsection{The variety approach to the decidability of $\FO^2_m[{<}]$}

Classes of finite monoids that are closed under taking submonoids, homomorphic images and finite direct products are called \emph{pseudovarieties}. The classes of finite monoids $\Jone$, $\Jtrivial$, $\Ap$, $\DA$, $\R{m}$ and $\L{m}$ introduced above are all pseudovarieties. 

If $\V$ is a pseudovariety of monoids, the class $\calV$ of languages recognized by a monoid in $\V$ is called a \emph{variety of languages}. Eilenberg's variety theorem (see \textit{e.g.}~\cite[Annex~B]{pp04}) shows that varieties of languages are characterized by natural closure properties, and that the correspondence $\V \mapsto \calV$ is onto. Elementary automata theory shows in addition that a language $L$ is recognized by a monoid in a pseudovariety $\V$ if and only the syntactic monoid of $L$ is in $\V$. It follows that if $\V$ has a decidable membership problem, then so does the corresponding variety of languages $\calV$.

Simon's Theorem on piecewise testable
languages~\cite{kli11dm:short,sim75} is an important instance of this Eilenberg correspondence: a language $L$ is recognizable by a monoid in $\Jtrivial$ if and only if $L$ is piecewise testable (and hence, as we already observed, if and only if $L$ is definable
in~$\FO^2_1[{<}]$). Simon's result implies the decidability of piecewise testability.

It immediately follows from the definition that
membership in $\R{m}$ and $\L{m}$ is decidable for all~$m$ since
membership in $\Jtrivial$ is decidable (see Corollary~\ref{corollary decidability} for a more precise statement). Many additional properties of the pseudovarieties $\R{m}$ and $\L{m}$, and
of the corresponding varieties of languages were established by the authors \cite{kw10sf,kw11LMCS,tw97au}. We will use in particular the following results, respectively \cite[Cor.~3.15]{kw11LMCS} and \cite[Thms.~2.1 and~3.5]{kw10sf}.

\begin{proposition}\label{generate Rm Lm}
An $A$-generated monoid $M$ is in $\R{m}$ (resp.\ $\L{m}$) if and only if
there exists an integer $n$ such that $M$ is a quotient of $A^* /
{\Right_{m,n}}$ (resp.\ $A^* / {\Left_{m,n}}$).
\end{proposition}

Let $x_1, x_2, \ldots$ be a sequence of variables. For each word $u$, we denote by $\bar u$ the mirror image of $u$, that is, the word obtained by reading $u$ from right to left. Let $G_2 = x_2x_1$, $I_2 = x_2x_1x_2$ and, for $m\ge 2$, $G_{m+1} = x_{m+1}\overline{G_m}$ and $I_{m+1} = G_{m+1}x_{m+1}\overline{I_m}$. Finally, let $\phi$ be the substitution given by
\begin{align*}
&\phi(x_1) = (x_1^\omega x_2^\omega x_1^\omega)^\omega,\quad \phi(x_2) = x_2^\omega, \\
&\textrm{and, for $m \ge 2$,}\quad\phi(x_{m+1}) = (x_{m+1}^\omega \phi(G_m\overline{G_m})^\omega x_{m+1}^\omega)^\omega.
\end{align*}

\begin{proposition}\label{identities Rm Lm}
$\R{m}$ (resp.\ $\L{m}$) is the class of finite monoids satisfying $(xy)^\omega x (xy)^\omega = (xy)^\omega$ and $\phi(G_m) = \phi(I_m)$ (resp.\ $\phi(\overline{G_m}) = \phi(\overline{I_m})$.
\end{proposition}

Straubing \cite{str11csl} and Kuf\-leitner and Lauser \cite[Cor. 3.4]{2012:KufleitnerLauser-fragments} established, by different means, that for each $m\ge 1$, the class of $\FO^2_m[{<}]$-definable languages forms a variety of languages, and we denote by $\VarFO^2_m$ the corresponding pseudovariety. In particular, $\VarFO^2_1 = \Jtrivial$. Our strategy to establish the decidability of $\FO^2_m[{<}]$-definability, is to establish the decidability of membership in $\VarFO^2_m$.

It is to be noted that neither Straubing's result, nor Kuf\-leitner's and Lauser's result implies the decidability of $\VarFO^2_m$. Straubing's result is the following \cite[Thm. 4]{str11csl}.

\begin{theorem}
For $m\ge 1$, $\VarFO^2_{m+1} = \VarFO^2_m \dast \Jtrivial$, where $\dast$ denotes the two-sided wreath product.
\end{theorem}

We refer the reader to \cite{str11csl} for the definition of the two-sided wreath product, which is also called the block product in the literature. As discussed by Straubing, this exact algebraic characterization of $\VarFO^2_m$ implies the decidability of $\VarFO^2_2$ but not of the other levels of the hierarchy. Straubing however conjectured that the following holds \cite[Conj.~10]{str11csl}.

\begin{conjecture}[Straubing] \label{straubingconjecture}
Let $u_1 = (x_1x_2)^\omega$, $v_1 = (x_2x_1)^\omega$ and, for $m\ge 1$,
\begin{align*}
u_{m+1} &= (x_1\cdots x_{2n}x_{2n+1})^\omega u_n (x_{2n+2}x_1\cdots x_{2n})^\omega \\
v_{m+1} &= (x_1\cdots x_{2n}x_{2n+1})^\omega v_n (x_{2n+2}x_1\cdots x_{2n})^\omega.
\end{align*}
Then a monoid is in $\VarFO^2_m$ if and only if it satisfies $x^{\omega+1} = x^\omega$ and $u_m = v_m$.
\end{conjecture}

\noindent If established, this conjecture would prove the decidability of each $\VarFO^2_m$.
The authors on the other hand proved the following \cite[Thm.~5.1]{kw11LMCS}.

\begin{theorem}\label{main result from LMCS}
If a language $L$ is
recognized by a monoid in the join $\R{m} \join \L{m}$, then $L$ is
definable in $\FO^2_m[{<}]$; and if $L$ is definable in
$\FO^2_m[{<}]$, then $L$ is recognized by a monoid in $\R{m+1} \cap
\L{m+1}$.
\end{theorem}

%%%%%%%%%%%%%%%%%%%%%%%
\section{The $\FO^2$ alternation hierarchy is decidable}\label{sec: main section}

We tighten the connection between the alternation hierarchy within $\FO^2[{<}]$ and the $\R{m}$-$\L{m}$-hierarchy and we prove the following result.

\begin{theorem}\label{thm:fo2m}
  A language $L \subseteq A^*$ is definable in $\FO^2_m[{<}]$ if and
  only if it is recognizable by a monoid in $\R{m+1} \cap \L{m+1}$.
\end{theorem}

Theorem~\ref{thm:fo2m} immediately yields a decidability result.

\begin{corollary}\label{corollary decidability}
For each $m\ge 1$, it is decidable whether a given regular language $L$ is $\FO^2_m[{<}]$-definable. This decision can be achieved in \textsc{Logspace} on input the multiplication table of the syntactic monoid of $L$, and in \textsc{Pspace} on input its minimal automaton.

Moreover, given a $\FO^2[{<}]$-definable language $L$, one can compute the least integer $m$ such that $L$ is $\FO^2_m[{<}]$.
\end{corollary}

\begin{proof}
We already observed that the $\R{m}$ and $\L{m}$ are decidable, and that each is described by two omega-term identities (Proposition~\ref{identities Rm Lm}). The decidability statement follows immediately. The complexity statement is a consequence of Straubing and Weil's~\cite[Thm.~2.19]{2012:StraubingWeil}. The computability statement follows immediately.
\qed
\end{proof}

We now turn to the proof of Theorem~\ref{thm:fo2m}. One implication was established in Theorem~\ref{main result from LMCS}. To prove the reverse implication, we prove \refprp{prp:main} below, which establishes that every language recognized by a monoid $M \in \R{m+1} \cap \L{m+1}$ is a union of $\equiv_{m,n}$-classes for some integer $n$ depending on $M$. Theorem~\ref{thm:fo2m} follows, in view of Theorem~\ref{WI 2009}.

\begin{proposition}\label{prp:main}
For every $m \geq 1$ and every morphism $\varphi\colon A^* \to M$ with $M \in \R{m+1} \cap \L{m+1}$ there exists an integer $n$ such that $\equiv_{m,n}$ is contained in $\equiv_\varphi$.
\end{proposition}

Before we embark in the proof of \refprp{prp:main}, we record several algebraic and combinatorial lemmas.
 
%%%%%%%%%%%%%%%%%%%%%%%
\subsection{A collection of technical lemmas}

\begin{lemma}\label{lem:lifteq}
  Let $M$ be a finite monoid. If $s \Req sx$ and $x \sim_{\K} y$, then
  $sx = sy$. If $s \Leq xs$ and $x \sim_{\D} y$, then $xs = ys$.
\end{lemma}

\begin{proof}
  Let $z \in M$ such that $sxz = u$. We have $(xz)^\omega x \Jeq
  (xz)^\omega$.  Now, $x \sim_{\K} y$ implies $(xz)^\omega x =
  (xz)^\omega y$. Thus $sx = s(xz)^\omega x = s(xz)^\omega y = sy$.
  The second statement is left-right symmetric.
  \qed
\end{proof}

The following lemma illustrates an important structural property of monoids in $\DA$.

\begin{lemma}\label{lem:da}
Let $\phi\colon A^* \to M$, with $M \in \DA$ and let $x,y,z \in A^*$ such that $\phi(x) \Req \phi(xy)$ and $\Alpha(z) \subseteq \Alpha(y)$. Then $\phi(x) \Req \phi(xz)$.
\end{lemma}

\begin{proof}
  The map $\Alpha\colon A^*\to \calP(A)$ can be seen as a morphism,
  where the product on $\calP(A)$ is the union operation. Since $M \in
  \DA$, we have $M/{\sim_\LI} \in \Jone$; let $\pi\colon M\to
  M/{\sim_\LI}$ be the projection morphism. It is easily
  verified that there exists a morphism $\psi\colon \calP(A)\to
  M/{\sim_\LI}$ such that $\psi\circ\Alpha = \pi\circ\phi$, see
  Figure~\ref{fig: lemma on DA}.
\begin{figure}\centering
    \begin{tikzpicture}[scale=0.1]
      \draw (0.0,-3.0) node (n0) {$A^*$};
      \draw (30.0,-3.0) node (n1) {$M$};
      \draw (0.0,-20.0) node (n3) {$\calP(A)$};
      \draw (30.0,-20.0) node (n4) {$M/{\sim_\LI}$};

      \draw[->] (n0) -- node[above] {$\phi$} (n1);
      \draw[->] (n1) -- node[right] {$\pi$} (n4);
      \draw[->] (n0) -- node[left] {$\Alpha$} (n3);
      \draw[->] (n3) -- node[below] {$\psi$} (n4);
    \end{tikzpicture}
    \caption{$M \in \DA = \LI\malcev \Jone$}\label{fig: lemma on DA}
\end{figure}

By assumption, $\phi(x) = \phi(xyt)$ for some $t\in A^*$, and hence
$\phi(x) = \phi(x)\phi(yt)^\omega$. Since $\Alpha((yt)^\omega) =
\Alpha((yt)^\omega z(yt)^\omega)$, we have $\phi(yt)^\omega \sim_\LI
\phi(yt)^\omega \phi(z) \phi(yt)^\omega$. Applying the definition of
$\sim_\LI$ with $e = f = \phi(yt)^\omega$, it follows that
$\phi(yt)^\omega = \phi(yt)^\omega \phi(z) \phi(yt)^\omega$ and we now
have
\begin{equation*}
  \phi(x) = \phi(x)\phi(yt)^\omega = \phi(x) \phi(yt)^\omega \phi(z) \phi(yt)^\omega = \phi(x) \phi(z) \phi(yt)^\omega.
\end{equation*}
Therefore $\phi(x) \Req \phi(xz)$, which concludes the proof.
\qed
\end{proof}

A proof of the following lemma can be found in~\cite[Prop.~3.6 and
Lem.~3.7]{kw11LMCS}.

\begin{lemma}\label{lem:combi}
  Let $m \geq 2$, $u,v \in A^*$, $a\in A$.
  \begin{enumerate}
  \item\label{bbb:combi} If $u \Right_{m,n} v$ and $u = u_{-} a u_{+}$
    and $v = v_{-} a v_{+}$ are $a$-left factorizations, then $u_{-}
    \Right_{m,n-1} v_{-}$ and $u_{+} \Right_{m,n-1} v_{+}$.
  \item\label{ccc:combi} If $u \Right_{m,n} v$ and $u = u_{-} a u_{+}$
    and $v = v_{-} a v_{+}$ are $a$-right factorizations, then $u_{-}
    \Right_{m,n-1} v_{-}$ and $u_{+} \Left_{m-1,n-1} v_{+}$.
  \end{enumerate}
  Dual statements hold for $u \Left_{m,n} v$.
\end{lemma}

\begin{lemma}\label{lem:1ranker}
  Let $m,n \geq 2$ and let $u = u_{-} a u_{+}$ and $v = v_{-} a v_{+}$
  be $a$-left factorizations.  If $u \equiv_{m,n} v$, then $u_{-}
  \equiv_{m-1,n-1} v_{-}$ and $u_{+} \equiv_{m,n-1} v_{+}$. A dual
  statement holds for the factors of the $a$-right factorizations of
  $u$ and $v$.
\end{lemma}

\begin{proof}
We first show $u_{-} \equiv_{m-1,n-1}
  v_{-}$.  Consider a ranker $r \in R_{m-1,n-1}$, supposing first that $r
  \in R^{\X}_{m-1,n-1}$. Then $r$ is defined on $u_{-}$ if and only if
  $r$ is defined on $u$ and $\ord(r'(u),\X_a(u))$ is
  ${<}$ for every nonempty prefix $r'$ of $r$. By definition of $\equiv_{m,n}$, this is equivalent to $r$ being defined on $v_-$. If instead $r \in R^{\Y}_{m-1,n-1}$, then $r$ is defined on $u_{-}$ if and
  only if $\X_a r \in R_{m,n}$ is defined on $u$ and $\ord(\X_a
  r'(u),\X_a(u))$ is ${<}$ for every nonempty prefix $r'$ of $r$. Again, this is equivalent to $r$ being defined on $v_-$ since $u\equiv_{m,n} v$. Thus, the same rankers in
  $R_{m-1,n-1}$ are defined on $u_{-}$ and $v_{-}$.

Now consider rankers $r \in R^{\X}_{m-1,n-1}$ and $s \in
  R^{\Y}_{m-1,n-2}$, which we can assume to be defined on both $u_{-}$ and $v_{-}$. Then the order types
  induced by $r$ and $s$ on $u_{-}$ and $v_{-}$ are equal, since
  $\ord(r(u_{-}),s(u_{-})) = \ord(r(u),\X_a s(u)) = \ord(r(v),\X_a
  s(v)) = \ord(r(v_{-}),s(v_{-}))$ and $\X_as \in R^{\X}_{m,n-1}$.
  
The same reasoning applies if $r \in R^{\Y}_{m-1,n-1}$ and $s\in R^{\X}_{m-1,n-2}$ (resp.\ if $r \in R^{\X}_{m-1,n-1}$ and $s\in R^{\X}_{m-1,n-2}$, if $r \in R^{\Y}_{m-1,n-1}$ and $s \in R^{\Y}_{m-2,n-2}$) since in that case, $\ord(r(u_{-}),s(u_{-})) = \ord(\X_ar(u),s(u))$ (resp.\ $\ord(r(u),s(u))$, $\ord(\X_a r(u),\X_a s(u))$).
Therefore, $u_{-} \equiv_{m-1,n-1} v_{-}$.

We now verify that $u_{+}
  \equiv_{m,n-1} v_{+}$. The proof is very similar to the first part
  and deviates only in technical details.
  Consider a ranker $r \in R_{m,n-1}$, say, in $R^{\X}_{m,n-1}$. Then $r$ is defined on $u_{+}$ if and only if $\X_a
  r \in R_{m,n}$ is defined on $u$ and $\ord(\X_a r'(u),\X_a(u))$
  is~${>}$ for every nonempty prefix $r'$ of $r$. Again, this is equivalent to $r$ being defined on
  $v_{+}$ since $u \equiv_{m,n} v$. If instead $r \in R^{\Y}_{m,n-1}$, then $r$ is defined on
  $u_{+}$ if and only if $r$ is defined on $u$ and
  $\ord(r'(u),\X_a(u))$ is~${>}$ for every nonempty prefix $r'$ of
  $r$, which is equivalent to $r$ being defined on $v_+$. Thus, the same rankers in $R_{m,n-1}$ are defined on $u_{+}$ and $v_{+}$.

Now consider rankers $r \in R^{\X}_{m,n-1}$ and $s \in R^{\Y}_{m,n-2}$, both defined on $u_{+}$ and $v_{+}$. Then the order types induced by $r$ and $s$ on $u_{+}$ and $v_{+}$ are equal, since $\ord(r(u_{+}),s(u_{+})) = \ord(\X_a r(u),s(u))$ and $\X_ar \in R^{\X}_{m,n}$.

Again, a similar verification guarantees that the order types induced by $r$ and $s$ on $u_+$ and $v_+$ are equal also if $r \in R^{\Y}_{m,n-1}$ and $s \in R^{\X}_{m,n-2}$, or if $r \in R^{\X}_{m,n-1}$ and $s \in
  R^{\X}_{m-1,n-2}$, or if $r \in R^{\Y}_{m,n-1}$ and $s \in
  R^{\Y}_{m-1,n-2}$.  This shows $u_{+} \equiv_{m,n-1} v_{+}$ which completes the proof.
  \qed
\end{proof}

\begin{lemma}\label{lem:cross:r}
Let $m,n \geq 2$ and let $u = u_{-} a u_0 b u_{+}$ and $v = v_{-} a v_0 b v_{+}$ describe $b$-left and $a$-right factorizations (that is, $a \not\in \Alpha(u_0 b u_{+}) \cup \Alpha(v_0 b v_{+})$ and $b \not\in \Alpha(u_{-} a u_0) \cup \Alpha(v_{-} a v_0)$).  If $u \equiv_{m,n} v$, then $u_{0} \equiv_{m-1,n-1} v_{0}$.
\end{lemma}

\begin{proof}
A ranker $r \in R^{\X}_{m-1,n-1}$ is defined on $u_{0}$ if and only if
  $\Y_a r \in R_{m,n}$ is defined on $u$ and $\ord(\Y_a
  r'(u),\Y_a(u))$ is~${>}$ and $\ord(\Y_a r'(u),\X_b(u))$ is ${<}$ for
  every nonempty prefix $r'$ of $r$. Similarly, a ranker $r \in
  R^{\Y}_{m-1,n-1}$  is defined on $u_{0}$ if and only if
  $\X_b r \in R_{m,n}$ is defined on $u$ and $\ord(\X_b
  r'(u),\Y_a(u))$ is~${>}$ and $\ord(\X_b r'(u),\X_b(u))$ is ${<}$ for
  every nonempty prefix $r'$ of $r$. Thus, if $u \equiv_{m,n} v$, then the same rankers in $R_{m-1,n-1}$ are defined on $u_0$ and $v_0$.

Now consider rankers $r \in R^{\X}_{m-1,n-1}$ and $s \in
  R^{\Y}_{m-1,n-2}$ (resp.\ $r \in R^{\Y}_{m-1,n-1}$ and $s \in
  R^{\X}_{m-1,n-2}$), defined on both $u_{0}$ and $v_{0}$. Then $\ord(r(u_{0}),s(u_{0})) =
  \ord(\Y_a r(u), \X_b s(u))$ (resp.\ $\ord(\X_b r(u), \Y_a s(u))$. Since $u \equiv_{m,n} v$, $\Y_a r \in R^{\Y}_{m,n}$ and
  $\X_b s \in R^{\X}_{m,n_1}$ (resp.\ $\X_b r \in R^{\X}_{m,n}$ and
  $\Y_a s \in R^{\Y}_{m,n_1}$), the order types defined by $r$ and $s$ on $u_0$ and $v_0$ are equal.
  
If $m = 2$, we are done proving that $u_0 \equiv_{m-1,n-1} v_0$. We now assume that $m\ge 3$.
Let $r \in R^{\X}_{m-1,n-1}$ and $s \in
  R^{\X}_{m-2,n-2}$ (resp.\ $r \in R^{\Y}_{m-1,n-1}$ and $s \in
  R^{\Y}_{m-2,n-2}$) be defined on both $u_{0}$ and $v_{0}$. Then
  $\ord(r(u_{0}),s(u_{0})) = \ord(\Y_a r(u),\Y_a s(u))$ (resp.\ $\ord(\X_b r(u),\X_b s(u))$). By the same reasoning as above, the order type defined by $v$ on $u_0$ and $v_0$ is the same since $\Y_a r \in
  R^{\Y}_{m,n}$ and $\Y_a s \in R^{\Y}_{m-1,n-1}$ (resp.\ $\X_b r \in
  R^{\X}_{m,n}$ and $\X_b s \in R^{\X}_{m-1,n-1}$).
This concludes the proof of the lemma.
  \qed
\end{proof}

%%%%%%%%%%%%%%%%%%%%%%%
\subsection{Proof of Proposition~\ref{prp:main}}

The proof is by induction on $m$. We already observed that $L$ is $\FO^2_1[{<}]$-definable if and only if it is piecewise testable, if and only if it is accepted by a monoid in $\J$. Since $\J = \R2\cap\L2$, Proposition~\ref{prp:main} holds for $m=1$. We now assume that $m\ge 2$.

Let $\phi\colon A^* \to M$ be a morphism with $M \in \R{m+1}\cap\L{m+1}$. We note that it suffices to prove Proposition~\ref{prp:main} for the morphism $\phi'\colon A^* \to M\times 2^A$ given by $\phi'(u) = (\phi(u),\Alpha(u))$. Observe that, for $u,v\in A^*$,
\begin{align}
& \phi'(u) \sim_\D \phi'(v) \textrm{ (resp.\ $\phi'(u) \sim_\K \phi'(v)$)}\quad\textrm{implies} \quad \Alpha(u) = \Alpha(v).
\end{align}
Indeed we have $\phi'(u)\phi'(u)^\omega = \phi'(u)^\omega$ (since $M$ is aperiodic): then $\phi'(u) \sim_\D \phi'(v)$ implies that $\phi'(v)\phi'(u)^\omega = \phi'(u)\phi'(u)^\omega$ and by definition of $\phi'$, $\Alpha(v)$ is contained in $\Alpha(u)$. By symmetry, $u$ and $v$ have the same alphabetical content and the same holds for $\sim_\K$.

To lighten up the notation, we dispense with the consideration of $\phi'$ and we assume that $\phi$ satisfies Property (1).

Let $\pi_{\D} :
  M \to M / {\sim_{\D}}$ and $\pi_{\K} : M \to M / {\sim_{\K}}$ be the
  natural morphisms. By definition of $\R{m+1}$ and $\L{m+1}$, we have $M / {\sim_{\D}} \in \R{m}$ and $M/{\sim_{\K}} \in \L{m}$. Let $\rho = \pi_{\D} \,\circ\,
  \varphi$ and $\lambda = \pi_{\K} \,\circ\,
  \varphi$, see Figure~\ref{fig: a commutative diagram}.  The monoid $A^* / ({\equiv_{\rho}} \join {\equiv_{\lambda}})$ is a
  quotient of both $M / {\sim_{\D}}$ and $M / {\sim_{\K}}$, so
  $A^* / ({\equiv_{\rho}} \join {\equiv_{\lambda}}) \in \R{m} \cap
  \L{m}$ and there exists $n \geq 1$ such that
  \begin{itemize}
  \item $\Right_{m,n}$ is contained in $\equiv_\rho$ and $\Left_{m,n}$ is contained in $\equiv_\lambda$ (by Proposition~\ref{generate Rm Lm}),
  \item $\equiv_{m-1,n}$ is contained in ${\equiv_{\rho}} \join
    {\equiv_{\lambda}}$ (by induction).
  \end{itemize}
\begin{figure}[t]  \begin{center}
    \begin{tikzpicture}[scale=0.75]
      \draw (0,0) node (A) {$A^*$};
      \draw (0,-2) node (M) {$M$};
      \draw (-2,-4) node (MD) {$M / {\sim_{\D}}$};
      \draw (2,-4) node (MK) {$M / {\sim_{\K}}$};
      \draw (0,-6) node (MDK) {$A^* / ({\equiv_{\rho}} \join {\equiv_{\lambda}})$};
      \draw (-5,-4) node (ARmn) {$A^* / {\Right_{m,n}}$};
      \draw (5,-4) node (ALmn) {$A^* / {\Left_{m,n}}$};
      \draw (7.5,-6) node (AEmn) {$A^* / {\equiv_{m-1,n}}$};

      \draw[->] (A) -- node[right] {$\varphi$} (M);
      \draw[->] (A) .. controls (-2,-2.5) .. node[left] {$\rho$} (MD);
      \draw[->] (A) -- (ARmn);
      \draw[->] (A) .. controls (2,-2.5) .. node[right] {$\lambda$} (MK);
      \draw[->] (A) -- (ALmn);
      \draw[->] (M) -- node[pos=0.35,left] {$\pi_{\D}$} (MD);
      \draw[->] (M) -- node[pos=0.35,right] {$\pi_{\K}$} (MK);
      \draw[->] (ARmn) -- (MD);
      \draw[->] (ALmn) -- (MK);
      \draw[->] (MD) -- (MDK);
      \draw[->] (MK) -- (MDK);
      \draw[->] (A) .. controls (7,-4) .. (AEmn);
      \draw[->] (AEmn) -- (MDK);
    \end{tikzpicture}
  \end{center}
  \caption{A commutative diagram}\label{fig: a commutative diagram}
\end{figure}
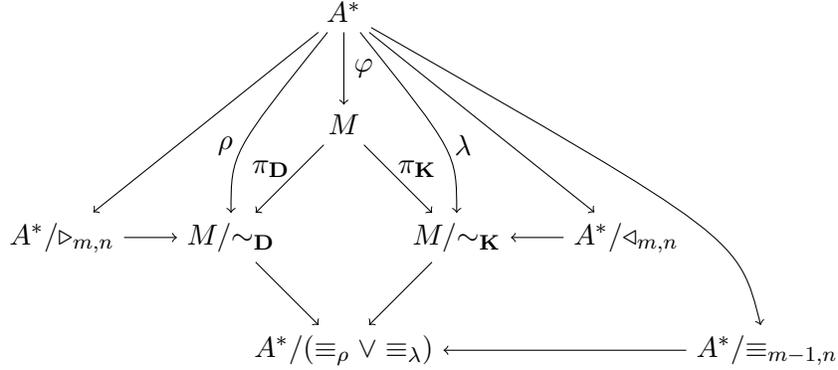

We show that $\equiv_{m,n+2\abs{M}}$ is contained in
  $\equiv_\varphi$. Let $u \equiv_{m,n+2\abs{M}} v$.
  Consider the $\greenR$-factorization of $u$, i.e., $u = s_1 a_1
  \cdots s_k a_k s_{k+1}$ with $a_i \in A$ and $s_i \in A^*$ such that
  $1 = \varphi(s_1)$ and for all $1 \leq i \leq k$:
  \begin{equation*}
    \varphi(s_1 a_1 \cdots s_i) 
    \Rg \varphi(s_1 a_1 \cdots s_i a_i) 
    \Req \varphi(s_1 a_1 \cdots s_i a_i s_{i+1}).
  \end{equation*}
  Since the number of $\greenR$-classes is at most $\abs{M}$, we have
  $k < \abs{M}$.  Similarly, let $v = t_1 b_1 \cdots t_{k'} b_{k'}
  t_{k'+1}$ with $b_i \in A$ and $t_i \in A^*$ be the
  $\greenL$-factorization of $v$ such that $\varphi(t_{k'+1}) = 1$ and
  for all $1 \leq i \leq k'$:
  \begin{equation*}
    \varphi(t_{i} b_i t_{i+1} \cdots b_{k'} t_{k'+1})
    \Leq \varphi(b_i t_{i+1} \cdots b_{k'} t_{k'+1})
    \Ll \varphi(t_{i+1} \cdots b_{k'} t_{k'+1}).
  \end{equation*}
  As before, we have $k' < \abs{M}$. By \reflem{lem:da} (applied with $x = s_1\cdots s_{i-1}a_{i-1}$, $y = s_i$ and $z = a_i$), we have $a_i
  \not\in \Alpha(s_i)$; and similarly, $b_i \not\in \Alpha(t_{i+1})$. Therefore,
  the positions of the $a_i$'s in $u$ are exactly the positions
  visited by the ranker $r = \X_{a_1} \cdots \X_{a_k}$, and the
  positions of the $b_i$'s in $v$ are exactly the positions visited by
  the ranker $s = \Y_{b_{k'}} \cdots \Y_{b_1}$. Since $u
  \equiv_{m,n+2\abs{M}} v$, each of the rankers $r$ and $s$ is defined
  on both $u$ and $v$, and all the positions visited by the rankers $r$
  and $s$ occur in the same order in $u$ as in $v$. We call these positions \emph{special}. Let
  \begin{align*}
    u &= u_1 c_1 \cdots u_\ell c_\ell u_{\ell+1} \\
    v &= v_1 c_1 \cdots v_\ell c_\ell v_{\ell+1}
  \end{align*}
  be obtained by factoring $u$ and $v$ at all the special positions. We have $\ell \leq k + k' < 2\abs{M}$. We say that a special position is
  \emph{red} if it is visited by $r$, and that it is \emph{green} if it is visited by
  $s$. Some special positions may be both red and green, which means that more
  than one of the cases below may apply.

  For $u$ the above factorization is a refinement of the
  $\greenR$-factorization; and for $v$ it is a refinement of the
  $\greenL$-factorization. In particular, $\varphi(u_1) = 1$,
  $\varphi(v_{\ell+1})=1$ and
  \begin{align}
    \varphi(u_1 \cdots u_{i-1} c_{i-1}) 
    &\Req \varphi(u_1 \cdots u_{i-1} c_{i-1} u_i) &&
    \textrm{for }1 < i \leq \ell+1, \tag{Eq($\Req$)}\\
    \varphi(v_i c_i v_{i+1} \cdots c_\ell)
    &\Leq \varphi(c_i v_{i+1} \cdots c_\ell) &&
    \text{for } 1 \leq i \leq \ell. \tag{Eq($\Leq$)}
  \end{align}
In order to prove $u \equiv_\phi v$, we show that we can gradually substitute $u_i$ for $v_i$ in the product $v_1c_1 \cdots v_\ell c_\ell v_{\ell+1} = v$, starting from $i = 1$, while maintaining $\equiv_\phi$-equivalence. Namely we show that, for each $i$, it holds
\begin{equation*}
u_1\cdots u_{i-1}c_{i-1} \, u_i \, c_iv_{i+1} \cdots v_{\ell+1} \equiv_\phi u_1\cdots u_{i-1}c_{i-1} \, v_i \, c_iv_{i+1} \cdots v_{\ell+1}.\tag{Eq(i)}
\end{equation*}

Let $h_0$ be the leftmost red position: then $c_{h_0} = a_1$ and $s_1 = u_1c_1\cdots u_{h_0}$. Since $\phi(s_1) = 1$ and $M$ is aperiodic, the $\phi$-image of every letter in $s_1$ is 1. Applying \reflem{lem:1ranker} to the $a_1$-left factorizations of $u$ and $v$, we find that $u_1c_1\cdots u_{h_0-1} \equiv_{m-1,n-1} v_1c_1 \cdots v_{h_0-1}$ and in particular, these words have the same alphabet. It follows that $\phi(u_i) = \phi(v_i) = 1$ for all $i \le h_0$, and hence (Eq($i$)) holds for all $i \le h_0$.

The right-left dual of this reasoning establishes that $\phi(u_i) = \phi(v_i) = 1$ for all the $u_i, v_i$ to the right of the last (rightmost) green position, say $j_0$. In particular, (Eq($i$)) also holds for all $i > j_0$.

We now assume that $h_0 < i \le j_0$ and we let $h-1$ be the first red position to the left of $i$ and $j$ be the first green position to the right of $i$: we have $h_0 < h \le i \le j\le j_0$.

\paragraph{Case 1: $h = i$ ($i-1$ is red)}
We have $u \Right_{m,n+2\abs{M}}
  v$.  By \reflem{lem:combi}\,(\ref{bbb:combi}), a
  sequence of at most $i-1$ left-factorizations yields $u_i c_i \cdots
  u_{\ell+1} \Right_{m,n+2\abs{M}-i+1} v_i c_i \cdots v_{\ell+1}$.  If
  $i$ is red, then by \reflem{lem:combi}\,(\ref{bbb:combi}), after
  one $c_i$-left-factorization, we see that $u_i
  \Right_{m,n+2\abs{M}-i} v_i$. If $i$ is not red, then $i$ is green and by
  \reflem{lem:combi}\,(\ref{ccc:combi}), after at most $\ell-i$
  right-factorizations, we find that $u_i$ and $v_i$ are
  $\Right_{m,n+2\abs{M}-i-(\ell-i)}$-equivalent. In any case, we have $u_i
  \Right_{m,n} v_i$ and thus $u_i \equiv_\rho v_i$ (i.e., $\phi(u_i) \sim_\D \phi(v_i)$) by the choice of
  $n$. In view of (Eq($\Leq$)), \reflem{lem:lifteq} now implies
\begin{equation*}
  u_i c_i v_{i+1} \cdots c_\ell v_{\ell + 1}\equiv_\varphi v_i c_i v_{i+1} \cdots c_\ell v_{\ell + 1}
\end{equation*}
and left multiplication by $u_1c_1\cdots c_{i-1}$ yields (Eq($i$)).

\paragraph{Case 2: $j=i$ ($i$ is green)}
As in Case~1, we see that $u_i
  \equiv_\lambda v_i$. (Eq($\Req$)) and \reflem{lem:lifteq} then imply

\begin{equation*}
  u_1 c_1 \cdots u_{i-1} c_{i-1} u_i \equiv_\varphi u_1 c_1 \cdots u_{i-1} c_{i-1} v_i,
\end{equation*}
and right multiplication by $c_iv_{i+1}\cdots v_{\ell+1}$ yields (Eq($i$)).

\paragraph{Case 3: $h < i < j$ ($i-1$ is not red and $i$ is not green)}
By \reflem{lem:1ranker}, after at most $h-1$ left factorizations and
  $\ell-j+1$ right factorizations, we obtain $u_h c_h \cdots u_j
  \equiv_{m,n+j-h} v_h c_h \cdots v_j$ (since $n+j-h \leq n +
  2\abs{M} - (h - 1) - (\ell - j + 1)$). \reflem{lem:cross:r}, applied with
  $a = c_{i-1}$ and $b = c_i$, then yields $u_i \equiv_{m-1,n} v_i$.  Since $\equiv_{m-1,n}$ is contained in $\equiv_\lambda\vee\equiv_\rho$, there exist words $w_1, \ldots, w_d$ such that
  \begin{equation*}
    v_i = w_1 \equiv_\rho w_2 \equiv_\lambda w_3 \equiv_\rho \cdots
    \equiv_\lambda w_{d-2} \equiv_\rho w_{d-1} \equiv_\lambda w_d = u_i.
  \end{equation*}
After the discussion at the beginning of this section, we have $\Alpha(v_i) = \Alpha(w_2) = \cdots =
  \Alpha(w_{d-1}) = \Alpha(u_i)$. Thus, by \reflem{lem:da}, we have
  $\varphi(p u_i) \Req \varphi(p)$ if and only if $\varphi(p w_g) \Req
  \varphi(p)$, and $\varphi(v_i q) \Leq \varphi(q)$ if and only if
  $\varphi(w_g q) \Leq \varphi(q)$ for all $p,q \in A^*$. As in Cases 1 and 2, we conclude that for each $1 \le e < d$,
\begin{itemize}
\item if $w_e \equiv_\rho w_{e+1}$, then
\begin{align*}
w_e c_i\cdots c_\ell v_{\ell+1} &\equiv_\phi w_{e+1} c_i\cdots c_\ell v_{\ell+1} \text{, and thus} \\
u_1c_1 \cdots u_i c_{i-1}w_e c_i\cdots c_\ell v_{\ell+1} &\equiv_\phi u_1c_1 \cdots u_i c_{i-1}w_{e+1} c_i\cdots c_\ell v_{\ell+1};
\end{align*}
\item and if $w_e \equiv_\lambda w_{e+1}$, then
\begin{align*}
u_1c_1 \cdots c_{i-1}w_e &\equiv_\phi u_1c_1 \cdots c_{i-1} w_{e+1}  \text{, and thus} \\
u_1c_1 \cdots c_{i-1}w_e c_i v_{i+1}\cdots c_\ell v_{\ell+1} &\equiv_\phi u_1c_1 \cdots c_{i-1}w_{e+1} c_i v_{i+1}\cdots c_\ell v_{\ell+1}.
\end{align*}
\end{itemize}
It follows by transitivity of $\equiv_\phi$ that (Eq($i$)) holds.

\paragraph{Concluding the proof}
We have now established (Eq($i$)) for every $1\le i \le \ell+1$. It follows immediately, by transitivity, that $u \equiv_\varphi v$.
  \qed

%%%%%%%%%%%%%%%%%%%%%%%
\section{Conclusion}

We have shown that for each $m\ge 1$, it is decidable whether a given regular language is $\FO^2_m[<]$-definable. Previous results in the literature only showed decidability for levels 1 and 2 of this quantifier alternation hierarchy. Our decidability result follows from the proof that $\VarFO^2_m$ (the pseudovariety of finite monoids corresponding to the $\FO^2_m[<]$-definable languages) is equal to the intersection $\R{m+1}\cap\L{m+1}$, which was known to be decidable.

This result implies the decidability of the levels of the hierarchy given by $\V_1 = \J$ and $\V_{m+1} = \V\dast\J$, since Straubing showed that $\V_m = \VarFO^2_m$ \cite{str11csl}. Straubing  used general results of Almeida and Weil on two-sided semidirect products to deduce from this that $\VarFO^2_2$ is decidable, but these results do not extend to $\VarFO^2_m$ when $m > 2$ (\cite{1998:AlmeidaWeil,2002:Weil}, see \cite[Sec. 5]{str11csl} for a discussion).

We also showed that the decision procedure whether a regular language $L$ is $\FO^2_m$-definable, is in \textsc{Logspace} on input the multiplication table of the syntactic monoid of $L$, and in \textsc{Pspace} on input the minimal automaton of $L$. The result behind this statement is the fact that membership in $\R{m}$ and in $\L{m}$ is characterized by a small set of (rather complicated) identities. Straubing conjectured a different and simpler set of identities (Conjecture~\ref{straubingconjecture} above). Our results do not confirm this conjecture, which it would be interesting to settle. 

{\small
\bibliographystyle{abbrv}

}
%\end{spacing}

%%%%%%%%%%%%%%%%%%%
%%%%%%%%%%%%%%%%%%%
%%%%%%%%%%%%%%%%%%%
\end{document}